\documentclass{article}
\usepackage{amsmath,amsthm,amssymb}
\usepackage{mathrsfs}
\usepackage[hidelinks]{hyperref}
\DeclareMathOperator*{\argmax}{arg\,max}

\def\cC{\mathscr{C}}
\def\cP{\mathcal{P}}
\def\R{\mathbb{R}}

\def\wV{\mathrm{wV}}

\newtheorem{thm}{Theorem}[section]
\newtheorem*{thm*}{Theorem}
\newtheorem{cor}[thm]{Corollary}
\newtheorem{lm}[thm]{Lemma}
\newtheorem{pp}[thm]{Proposition}
\newtheorem{ft}[thm]{Fact}

\theoremstyle{definition}
\newtheorem{df}[thm]{Definition}
\newtheorem{eg}[thm]{Example}
\theoremstyle{remark}
\newtheorem{rk}[thm]{Remark}
\newtheorem{st}[thm]{Statement}
\usepackage{orcidlink}

\providecommand{\keywords}[1]
{
	\small	
	\textbf{\textit{Keywords---}} #1
}

\title{Nash equilibria of quasisupermodular games}
\author{Lu Yu\thanks{Université Paris 1 Panthéon-Sorbonne, UMR 8074, Centre d'Economie de la Sorbonne, Paris, France, \href{mailto:yulumaths@gmail.com}{yulumaths@gmail.com}} \,\orcidlink{0000-0001-6154-4229}}
\date{\today}
\begin{document}\maketitle\begin{abstract}
		We prove three results on the existence and structure of Nash equilibria for quasisupermodular games. A theorem is purely order-theoretic, and the other two involve topological hypotheses.  Our topological results generalize Zhou's theorem (for supermodular games) and Calciano's theorem.
	\end{abstract}\keywords{Quasisupermodular games, Nash equilibrium, Tarski fixed point theorem}

\section{Introduction}
Quasisupermodular games, introduced by   Shannon and Milgrom in \cite{shannon1990ordinal,milgrom1994monotone},  are generalizations of supermodular games. Supermodular games exhibit ``strategic complementarities", wherein the choices of one player positively influence the preferences of others, leading them to favor similar higher actions in response.

The defining properties of supermodular games, namely supermodularity and  increasing difference property for \emph{real-valued} payoff functions, rely essentially on the additive group structure of the \emph{codomain} $\R$. Milgrom and Shannon \cite[p.160, p.162]{milgrom1994monotone} generalize the two conditions  to quasisupermodularity (Definition \ref{df:qsi}) and single crossing property (Definition \ref{df:single}). Their conditions require no longer  any additive structure on the codomain, so extend to chain-valued functions (\cite[p.59]{topkis1998supermodularity}). 

Assuming order upper semicontinuity on the payoff functions, Milgrom and Roberts  \cite[p.1266]{milgrom1990rationalizability} show that a supermodular game has a least and a largest Nash equilibrium.   Under (topological) continuity conditions   on the payoff functions, Milgrom and Shannon  \cite[Theorem 12]{milgrom1994monotone} give a similar result for quasisupermodular games.   The same conclusion holds even if the continuity restriction is relaxed and the payoff functions are allowed to be more general poset-valued functions (\cite[Proposition 8.49]{carl2010fixed}). Here and below, a set equipped with a partial order is called a \emph{poset}.

The study on the structure of Nash equilibria of supermodular games is initialized by Vives \cite[Theorem 4.2]{vives1990nash}. Veinott \cite[Ch.~10]{veinott1992lattice} establishes the existence and studies the structure of Nash equilibria for parameterized games with certain complementarities, which are variants of quasisupermodular games.  Zhou \cite[Theorem 2]{zhou1994set} shows that for a supermodular game, the Nash equilibrium set with the inherited order structure forms a complete lattice. Zhou's result is extended to quasisupermodular games by Calciano \cite[Theorem 3]{calciano2007games}.

The main objective of this paper is twofold: the first  concerns  (pure) Nash equilibria of quasisupermodular (normal form) games. We provide a  \emph{purely order-theoretic} result (Theorem \ref{thm:order}) concerning the structure of the set of Nash equilibria. As opposed to Nash's theorem on the existence of Nash equilibria, it does not require any topological hypothesis. The proof of Theorem \ref{thm:order} needs an existence result of maximum (Theorem \ref{thm:Kuku}).  Second, in the spirit of \cite[Theorem 1]{prokopovych2017strategic}, we relax the conventional topological condition (upper semicontinuity) of the payoff functions, to \emph{transfer upper continuity} (Definition \ref{df:transferctu}) in Theorems  \ref{thm:topo+order} and \ref{thm:topo}.

The text is organized as follows. Section \ref{sec:order} summarizes two completeness properties for lattices and shows their equivalence. Section \ref{sec:ctu} recalls several conditions analogous to continuity, both in the topological and the order-theoretic sense. Section \ref{sec:qs} reviews quasisupermodularity.  Section \ref{sec:max} gives existence results about maximum for functions. Finally in Section \ref{sec:qsgame}, results from the previous sections are used to establish the existence of Nash equilibria for quasisupermodular games. The order structure of the set of all Nash equilibria is also discussed.

\section{A result of Veinott}\label{sec:order}
In Section \ref{sec:order}, we give a new self-contained proof of a theorem due to Veinott \cite[Ch.~2]{veinott1992lattice} (recalled as Lemma \ref{lm:Veinott}). It is used in the proof of Theorem \ref{thm:order}, and we cannot find a published proof in the literature. Veinott's unpublished proof relies on transfinite induction.

A poset $L$ is called a \emph{lattice} if for any $x,y\in L$, both $\sup_L\{x,y\}$ and $\inf_L\{x,y\}$ exist in $L$, classically denoted by $x\vee y$ and $x\wedge y$ respectively. A poset $L$ is called a \emph{complete lattice} if for every nonempty subset $A$ of $L$, both $\sup_L(A)$ and $\inf_L(A)$ exist. (Such an $L$ is necessarily a lattice.) A poset is called a \emph{chain} (resp. an \emph{anti-chain}) if any two distinct elements are comparable (resp. incomparable). In an anti-chain with more than one element, every element is minimal but not the least element. For a lattice,  Lemma \ref{lm:min} shows the two notions agree.
\begin{lm}\label{lm:min}A minimal (resp.~maximal) element  of a lattice is the least (resp.~largest) element. \end{lm}\begin{proof}Let $s\in L$ be a minimal element of a lattice $L$. For every $s'\in L$, the element $s\wedge s'$ exists in $L$. Since $s$ is minimal and $s\wedge s'\le s$, one has $s=s\wedge s'\le s'$, which shows that $s$ is the least element. The statement in the parentheses is similar. \end{proof}
\begin{df}\label{df:subcompleteT2}
	Let $P$ be a poset.  Let $S$ be a subset of $P$. If for any $x,y\in S$, one has $\inf_P\{x,y\},\sup_P\{x,y\}\in S$, then $S$ is called a \emph{sublattice} of $P$.  If for every nonempty chain (resp. subset) $C\subset S$, both $\sup_P(C)$ and $\inf_P(C)$ exist and lie in $S$, then $S$ is called \emph{chain-subcomplete} (resp.~\emph{subcomplete})  in $P$. If $P$ is chain-subcomplete in itself, then  $P$ is called  \emph{chain-complete}.
\end{df}
The complete sublattice $[0,1)\cup\{2\}$ of the complete chain $[0,2]$ is not subcomplete. This example  can be found in \cite[p.1260]{milgrom1990rationalizability}.

In \cite[Theorem B]{kukushkin2012existence}, Lemma \ref{lm:Veinott} below is stated without proof  and attributed to Veinott's unpublished work \cite{veinott1992lattice}.  We provide a new proof for the reader's convenience. By \cite[p.115]{johnson1998elements}, for every set $X$, there is an object $|X|$, called the \emph{cardinal number} or \emph{cardinality} of $X$, with the property that for any two sets $X$ and $Y$, there is a bijection $X\to Y$ exactly when $|X|=|Y|$.
\begin{lm}[Veinott]\label{lm:Veinott}
	Let $S$ be a sublattice of a poset $P$. Then $S$ is subcomplete if and only if it is chain-subcomplete in $P$.
\end{lm}
\begin{proof}The ``only if" part follows from the definition. We prove the ``if" part.
	
	Suppose to the contrary that $S$ is not subcomplete. Then there is a nonempty subset $A\subset S$ such that $\sup_P(A)$ or $\inf_P(A)$ is not in $S$. By symmetry, we may assume that $\sup_P(A)$ is not in $S$. Let $\cP$ be the set of all nonempty subsets $A'$ of $S$ such that $\sup_P(A')$ does not exist in $S$. Then $A\in\cP$.
	From \cite[the first sentence]{honig1954proof}), we may assume  $|A|\le |A'|$ for all $A'\in\cP$.  Since $S$ is a lattice, the set $A$ is infinite. Then by \cite[Thm.~2.17]{roman2008lattices}, there is a chain $I$ and a nonempty subset $A_{\alpha}\subset A$ for every $\alpha\in I$, such that: 
	\begin{itemize}
		\item\label{it:Aalpha} For all $\alpha\in I$, one has $|A_{\alpha}|<|A|$;
		\item\label{it:Aunion} One has $A=\cup_{\alpha\in I}A_{\alpha}$;
		\item\label{it:incrsubset} For any $\alpha<\beta$ in $I$, one has $A_{\alpha}\subset A_{\beta}$.
	\end{itemize} 
	
	By \ref{it:Aalpha}, one has $A_{\alpha}\notin \cP$ for  all  $\alpha\in I$.
	Therefore, $b_{\alpha}:=\sup_P(A_{\alpha})$ exists  in $S$. By \ref{it:incrsubset}, for any  $\alpha\le \beta$ in $I$, one has  $b_{\alpha}\le b_{\beta}$. Thus,  $\{b_{\alpha}\}_{\alpha\in I}$ is a nonempty increasing chain in $S$. Since $S\subset P$ is chain-subcomplete, the element $b=\sup_P\{b_{\alpha}\}_{\alpha\in I}$ exists and is in $S$.
	
	We claim that $b=\sup_P(A)$. By \ref{it:Aunion}, for every $a\in A$, there is  $\alpha_0\in I$ with $a\in A_{\alpha_0}$. Then $a\le b_{\alpha_0}\le b$. If $b'\in P$ is another upper bound on $A$,  then $b'\ge b_{\alpha}$ for every $\alpha\in I$. Therefore, $b'\ge b$.  The claim is proved. However, the claim contradicts the choice of $A$. 
\end{proof}
\begin{rk}Lemma \ref{lm:Veinott} implies \cite[Thm.~2.41~(iii)]{davey2002introduction}: a lattice having no infinite chains is complete.\end{rk}
\section{Continuity}\label{sec:ctu}
In Section \ref{sec:ctu}, we review some variants   of the continuity in topology. They are different from the conventional upper semicontinuity, but suffice to ensure the existence of a maximum. They are part of the hypotheses of our main results in  Section \ref{sec:qsgame}.

Transfer continuity conditions are introduced by Tian and Zhou \cite[Definitions 1, 2]{tian1995transfer}. They are used to establish the existence of Nash equilibria in discontinuous games in \cite{nessah2016existence}.
\begin{df}\label{df:transferctu}
	Let $X$ be a topological space. A function $f:X\to \R$ is called \emph{transfer} (resp. \emph{weakly transfer}) \emph{upper continuous} if for any $x,y\in X$ with $f(y)<f(x)$, there exists a point $x'\in X$ and an open neighborhood  $U$ of $y\in X$, such that $f(z)<f(x')$ (resp. $f(z)\le f(x')$) for every $z\in U$.
\end{df}
An upper semicontinuous function is transfer upper continuous, 
and a transfer upper continuous is transfer weakly upper continuous. As \cite[Examples 1,2]{tian1995transfer} show, the converses fail.   

Let $\cC$ (resp.~$(L,\le)$) be a chain (resp.~complete lattice). The interval topology (in the sense of \cite[p.570]{frink1942topology}) of the lattice $L$ is defined  to be the smallest topology such that $\{x\in L:x\ge a\}$ and $\{x\in L:x\le a\}$ are closed for every $a\in L$.   It is a topology canonically determined by the order structure. 

We recall several sorts of continuity involving order.  Let $f:L\to \cC$ be a function. It is called \emph{topologically upper semicontinuous}, if for every $a\in \cC$, the set $[f\ge a]:=\{x\in L:f(x)\ge a\}$ is closed in the interval topology of $L$. (When $\cC=\R$, this reduces to the upper semicontinuity of real-valued functions.)

Definition \ref{df:orderctu} is  an order-theoretic counterpart of the topological semicontinuity.  We need the notation from  \cite[p.1261]{milgrom1990rationalizability}. For a nonempty chain $C\subset L$, the pair $(C,\le)$ is naturally a directed set. Then the inclusion $C\to L$ is a net in $L$, which is denoted by $x\in C,x\uparrow\sup_L(C)$. Define the reversed order $\preceq$ on $C$, by $x\preceq y$ if and only if $x\ge y$ for all $x,y\in C$. Then $(C,\preceq)$ is another directed set. The inclusion $C\to L$ gives another net in $L$, which is denote by $x\in C,x\downarrow\inf_L(C)$.  For every net $x_{\bullet}:(I,\le)\to L$, define the limit superior along the net by $\limsup_{i\in I}x_i:=\inf_{i\in I}\sup_{j\ge i}x_j$. It is well-defined by completeness of $L$.
\begin{df}[{\cite[p.1260]{milgrom1990rationalizability}, \cite[Definition 1]{prokopovych2017strategic}}]\label{df:orderctu}Let $f:L\to \cC$ be a  function  from a complete lattice to a complete chain. It is \emph{upward} (resp. \emph{downward}) \emph{upper semicontinuous}, if for every nonempty chain $C\subset L$,  one has \[\limsup_{x\in C,x\uparrow\sup_L(C)}f(x)\le f(\sup_LC)\]
	\[\text{(resp. }\limsup_{x\in C,x\downarrow\inf_L(C)}f(x)\le f(\inf_LC)).\] If $f$ is both upward and downward upper semicontinuous, then it is called \emph{order upper semicontinuous}.
\end{df}
\begin{rk}\label{rk:everychainisfinite}For a finite chain $C$, the corresponding inequalities in Definition \ref{df:orderctu} are automatic. In particular, if every chain in $L$ is finite, then every function on $L$ is order upper semicontinuous. The sum of two upward (resp. downward) upper semicontinuous functions on $L$ is also upward (resp. downward) upper semicontinuous.\end{rk}

Definition \ref{df:Shannon} is an extension of Definition \ref{df:orderctu} to non-complete chains. This notion is used in Theorem \ref{thm:order}.
\begin{df}[{\cite[p.8]{shannon1990ordinal}}]\label{df:Shannon} Let $f:L\to \cC$ be a  function  from a complete lattice to a  chain. It is \emph{upper chain subcomplete}, if for every $a\in \cC$ and every nonempty chain $C$ in $[f\ge a]$, both $\sup_L(C)$ and $\inf_L(C)$ are in $[f\ge a]$.
\end{df}
\begin{pp}\label{pp:usc}
	Let $f:L\to \cC$ be a function from a complete lattice to a complete chain. Then $f$ is order upper semicontinuous if and only if $f$ is upper chain subcomplete.
\end{pp}
\begin{proof}
	We prove that $f$ is upward upper semicontinuous if and only if for every $a\in \cC$ and every nonempty chain $C$ in $[f\ge a]$, one has $\sup_LC\in[f\ge a]$. 
	
	Assume that $f$ is upward upper semicontinuous. Then for every $a\in \cC$ and every nonempty chain $C$ in $[f\ge a]$, one has \[f(\sup_LC)\ge \limsup_{x\in C,x\uparrow\sup_LC}f(x)\ge a,\] so $\sup_LC\in [f\ge a]$.
	
	Conversely, assume that for every $a\in \cC$ and every nonempty chain $C'$ in $[f\ge a]$, one has $\sup_L(C')\in [f\ge a]$. For every nonempty chain $C\subset S$, set $b_2:=\limsup_{x\in C,x\uparrow\sup_LC}f(x)$.  We need to prove  \begin{equation}\label{eq:fsupSC}f(\sup_LC)\ge b_2.\end{equation} 
	Assume, by contradiction, $f(\sup_LC)<b_2$.   Because $\cC$ is a chain, there are exactly two cases.
	\begin{enumerate}
		\item\label{it:b1middle} There is $b_1\in \cC$ with $f(\sup_LC)<b_1<b_2$.
		\item\label{it:nointermediate}For every $\beta\in \cC$ with $\beta>f(\sup_LC)$, one has $\beta\ge b_2$. In this case, from $f(\sup_LC)<b_2$, for every $x\in C$, one has $\sup_{y\in C,y\ge x}f(y)>f(\sup_LC)$.  There is $y_x\in C$ with $y_x\ge x$ and $f(y_x)>f(\sup_LC)$. Then $f(y_x)\ge b_2$.
	\end{enumerate}

	In Case \ref{it:b1middle} (resp. \ref{it:nointermediate}), set $b$ to be $b_1$ (resp. $b_2$). Then in both cases, there is  a directed set $I$ and an order-preserving function $h:I\to C$ (i.e., a  Willard-subnet) such that $h(I)$ is a cofinal subset of $C$ and for every $u\in I$, one has $f(h(u))\ge b$. Then $h(I)\subset [f\ge b]$. 
	
	Since $h(I)\subset C$, one has $\sup_LC\ge \sup_L(h(I))$. For every $x\in C$, as $h(I)$ is cofinal in $C$, there is $y\in I$ with $x\le h(y)$. Then $x\le \sup_L(h(I))$ for all $x\in C$. Therefore, $\sup_LC\le \sup_L(h(I))$. Hence    $\sup_LC= \sup_L(h(I))$.

	As $h(I)$ is a chain in $[f\ge b]$, one has $\sup_LC= \sup_L(h(I))\in [f\ge b]$, i.e., \[f(\sup_LC)\ge b>f(\sup_LC).\] This contradiction proves \eqref{eq:fsupSC}. Then $f$ is upward upper semicontinuous.
	
	Similarly, $f$ is downward upper semicontinuous if and only if  for every $a\in \cC$ and every chain $C$ in $[f\ge a]$, one has $\inf_LC\in [f\ge a]$. The two parts finish the proof.
\end{proof}
Proposition \ref{pp:topo} shows that the topological continuity is stronger than the order continuity. 
\begin{pp}\label{pp:topo}	Let $f:L\to \cC$ be a  topologically upper semicontinuous function from a complete lattice to a  chain. Then $f$ is upper chain subcomplete. If further $\cC$ is complete, then $f$ is order upper semicontinuous.
\end{pp}
\begin{proof}
	For every $a\in \cC$, by the topological upper semicontinuity of $f$, the subset $[f\ge a]\subset L$  is closed. By Corollary \ref{cor:chaincls}, for every chain $C$ in $[f\ge a]$, one has $\sup_L(C),\inf_L(C)\in [f\ge a]$. Therefore, $f$ is upper chain subcomplete. When $\cC$ is complete, by Proposition \ref{pp:usc}, $f$ is order upper semicontinuous.
\end{proof}
To prove Corollary \ref{cor:chaincls} (used in the proof of Proposition \ref{pp:topo}), we need Lemma \ref{lm:chsup}. 
\begin{lm}\label{lm:chsup}
	Let $W,V$ be two nonempty subsets of   $L$ such that $C:=W\cup V$ is a chain. Then $\sup_L(C)=\sup_L(W)$ or  $\sup_L(C)=\sup_L(V)$.
\end{lm}
\begin{proof}
	Since $W\subset C$ (resp. $V\subset C$), one has $\sup_L(W)\le \sup_L(C)$ (resp. $\sup_L(V)\le \sup_L(C)$). Hence $\sup_L(W)\vee \sup_L(V)\le \sup_L(C)$. Since $\sup_L(W)\vee \sup_L(V)$ is an upper bound on $C$, one has $\sup_L(W)\vee \sup_L(V)=\sup_L(C)$.   If for every $u\in W$, there is $v\in V$ with $u\le v$, then $\sup_L(W)\vee \sup_L(V)=\sup_L(V)$. Otherwise, as $C$ is a chain, there is $u_0\in W$ such that $u_0\ge v$ for all $v\in V$. Hence $u_0\ge\sup_L(V)$ and $\sup_L(W)\vee \sup_L(V)=\sup_L(W)$.
\end{proof}
Roughly, Corollary \ref{cor:chaincls} below means that a subset closed in the interval topology is also closed under certain order theoretic operations.
\begin{cor}\label{cor:chaincls}
	Let $A\subset L$ be a subset closed in the interval topology of $L$. Then $A$ is chain-subcomplete in $L$.  
\end{cor}
\begin{proof}
	By the definition of  interval topology, one can write $A=\cap_{\alpha\in I}A_{\alpha}$, where  $A_{\alpha}$ is a finite union of closed intervals for each index $\alpha\in I$. More precisely, there is an integer $n(\alpha)\ge 1$ and closed intervals $I_{\alpha}^i$ ($i=1,\dots,n(\alpha)$) in $L$ with 
	$A_{\alpha}=\cup_{i=1}^{n(\alpha)}I_{\alpha}^i$. 	Let $C$ be a nonempty chain in $A$. For every $1\le i\le n(\alpha)$, one has $\sup_L(I_{\alpha}^i\cap C)\in I_{\alpha}^i$.
	For every $\alpha\in I$, one has $C\subset A_{\alpha}$, so $C=\cup_{i=1}^{n(\alpha)}(I_{\alpha}^i\cap C)$. By Lemma \ref{lm:chsup}, there is $1\le i_{\alpha}\le n(\alpha)$ with  $\sup_L(C)=\sup_L(I_{\alpha}^{i_{\alpha}}\cap C)$. For every $\alpha\in I$, one has $\sup_L(C)\in I_{\alpha}^{i_{\alpha}}\subset A_{\alpha}$. Hence $\sup_LC\in A$. Similarly, one has $\inf_LC\in A$. 
\end{proof}
\begin{rk}
	In Example \ref{eg:A3}, the subset $X\subset L$ is chain-subcomplete, but not closed in the interval topology. Whence, the converse of Corollary \ref{cor:chaincls} fails. Topkis \cite[Thm.~2.3.1]{topkis1998supermodularity} shows that for every integer $n\ge0$, every chain-subcomplete \emph{sublattice} of $\R^n$ is closed in the \emph{Euclidean} topology. Still, the \emph{subset} $\{(t,1-t)|0<t<1\}\subset \R^2$ is chain-subcomplete,  but not closed in the Euclidean topology. 
\end{rk}

\section{Quasisupermodular functions}\label{sec:qs}Let $L$ be a lattice. A function $g:L\to \R$  is called \emph{supermodular} if for any $x,y\in L$, one has $g(x)+g(y)\le g(x\wedge y)+g(x\vee y)$. The condition uses the additive group structure of the codomain $\R$. Thus, supermodularity  may not be preserved by a strictly increasing transformation (\cite[Example 2.6.5]{topkis1998supermodularity}). Quasisupermodularity is a  generalization of supermodularity that is preserved by such a transformation.  
\begin{df}\cite[p.162]{milgrom1994monotone}\label{df:qsi}
	A function $f:L\to C$ to a chain $C$ is called 
	\emph{quasisupermodular} if for any $x,y\in L$, the condition $f(x)\ge f(x\wedge y)$ implies $f(x\vee y)\ge f(y)$, and the condition $f(x)> f(x\wedge y)$ implies $f(x\vee y)> f(y)$.
\end{df}

Every supermodular function is quasisupermodular. Example \ref{eg:sumnotquasisupermodular} shows that unlike supermodular functions, the sum of two quasisupermodular function may not be quasisupermodular.
\begin{eg}\label{eg:sumnotquasisupermodular}
	Consider the sublattice $L=\{(1,1),(2,3),(3,2),(4,5)\}$ of $\R^2$. The function $f:L\to \R,\quad (x,y)\mapsto x$  is supermodular. The function $g:L\to \R,\quad (x,y)\mapsto -y$  is quasisupermodular. But the sum $h:=f+g$ is not quasisupermodular, since $h((2,3)\wedge (3,2))=h(1,1)<h(3,2)$ and $h(2,3)=h(4,5)=h((2,3)\vee (3,2))$.
\end{eg}

\begin{st}\cite[Theorem A3]{milgrom1994monotone}\label{st:A3}Let $L$ be a complete lattice. Let $f:L\to \R$ be a quasisupermodular and order upper semicontinuous function. Then $f$ is  upper semicontinuous in the interval topology of $L$.
\end{st}
However, Statement \ref{st:A3} is false even if $f$ is  supermodular.  Example \ref{eg:A3} is a counterexample. (Kukushkin \cite[Example 3.2]{kukushkin2009existence} provides a counterexample to \cite[Theorem A2]{milgrom1994monotone}.)
\begin{eg}\label{eg:A3}
	Let $X$ be an infinite anti-chain.	Let $L$ be the set obtained by adding two elements $m\neq M$ to $X$. We extend the partial order from $X$ to $L$ by requiring $m\le x\le M$ for every $x\in X$. Then $L$ is a complete lattice.  The interval topology of $L$ is $T_1$, compact, but not Hausdorff. 
	
	Fix $x_0\in X$. Define a function $f:L\to \R$  by setting $f(x)=1$ for every $x\neq x_0$ and $f(x_0)=0$.  By Remark \ref{rk:everychainisfinite}, since every chain in $L$ is finite,  $f$ is supermodular and order upper semicontinuous. Moreover, $\argmax f$ is compact but not closed in the interval topology of $L$. In particular, $f$ is not upper semicontinuous in the interval topology. (Thus, the converse of Proposition \ref{pp:topo} fails.)

	Take an element $x'(\neq x_0)$ of $L$. Then $f(x_0)<f(x')$ but $f$ does not satisfy the defining property of transfer upper continuity. This appears to contradict the literal statement of \cite[Theorem 2]{tian1995transfer}.  The reason is that the  topology in \cite[Theorem 2]{tian1995transfer} is implicitly assumed   to be Hausdorff.
\end{eg}

\section{Existence of  maximum}\label{sec:max}
In a normal form game, every player wants to maximize his/her own payoff function within the set of available  strategies. In this sense, the study of existence of a maximum of a given function is meaningful. The main objective of Section \ref{sec:max} is Theorem \ref{thm:Kuku}, which is used in Section \ref{sec:qsgame}.

Fact \ref{ft:Tian} is \cite[Theorem 2]{tian1995transfer}, except that the authors implicitly assume $X$  to be Hausdorff, as  Example \ref{eg:A3}  explains. Since the interval topology  of a complete lattice may not be Hausdorff, we state the following version.  The original proof still works.
\begin{ft}[Tian-Zhou]\label{ft:Tian}
	Let $X$ be a compact (but not necessarily Hausdorff) topological space. Let  $f:X\to \R$ be a transfer upper continuous function. Then $\argmax f$ is a nonempty closed subset of $X$, hence compact.
\end{ft}

Lemma \ref{lm:Sha3} generalizes \cite[Propostion 3]{shannon1990ordinal}, which only asserts the completeness of $\argmax f$.
\begin{lm}\label{lm:Sha3}Let $f:L\to C$ be an  upper chain subcomplete function from a complete lattice to a chain. 
	If $\argmax f$ is a nonempty sublattice of $L$, then it is \emph{subcomplete} in $L$.  In particular, $\argmax f$ is  compact in the interval topology of $L$.
\end{lm}
\begin{proof}
	For every nonempty chain $\cC\subset \argmax f$, since $f$ is upper chain subcomplete and $\argmax f=\{s\in L:f(s)\ge \max f\}$, one has $\sup_L\cC,\inf_L\cC\in \argmax f$. Therefore, $\argmax f$ is chain-subcomplete in $L$. By Lemma \ref{lm:Veinott}, the sublattice $\argmax f$ is  subcomplete in $L$. The compactness follows from \cite[Topkis's theorem]{yu2022existence}.
\end{proof}
\begin{rk}\label{rk:durieu}In Lemma \ref{lm:Sha3},  the compactness is not clear \textit{a priori}. If the condition ``upper chain subcomplete" of   is changed to  ``upper semicontinuity in the interval topology of $L$", then  $\argmax f$  is \emph{compact} in $L$. \emph{Given} the compactness, from completeness of $L$ and \cite[Footnote d, p.187]{durieu2008ordinal}, this sublattice is subcomplete. Still, Example \ref{eg:A3} and Proposition \ref{pp:usc} show that an upper chain subcomplete function may \emph{not} be upper semicontinuous in the interval topology.\end{rk}
Theorem \ref{thm:Kuku}  generalizes simultaneously \cite[Theorems 1,~2]{milgrom1990rationalizability} and \cite[Thm.~A4]{milgrom1994monotone}. (Although the proof of \cite[Thm.~A4]{milgrom1994monotone} uses the wrong Statement \ref{st:A3}, the statement itself is correct.) It relaxes the supermodularity condition to quasisupermodularity and strengthens the completeness result to subcompleteness. It gives a purely order-theoretic sufficient condition for the existence of maximum of functions on lattices. (A related  condition is given in \cite[Corollaries 3.1,~3.2]{kukushkin2012existence}.)
\begin{thm}\label{thm:Kuku}Let $f:L\to C$ be a quasisupermodular,  upper chain subcomplete function from a complete lattice to a chain. 
	Then $\argmax f$ is a nonempty  \emph{subcomplete} sublattice of $L$.
\end{thm}
\begin{proof}
	Define a correspondence $F:f(L)\to 2^L$ by $F(c)=\{x\in L:f(x)\ge c\}$. Since $f$ is upper chain subcomplete, for every $c\in f(L)$, the value $F(c)$ is chain-subcomplete in $L$.
	
	For any $c<c'$ in $C$, every $x\in F(c)$ and every $x'\in F(c')$ with $x\wedge x'\notin F(c)$, one has $f(x\wedge x')<c\le f(x)$. Since $f$ is quasisupermodular, one has $c'\le f(x')<f(x\vee x')$. Thus, $x\vee x'\in F(c')$. Therefore, the correspondence  $F$ is increasing with respect to the relation $\ge^{\wV}$ defined in \cite[(2d)]{kukushkin2013increasing}. By \cite[Theorem 2.2]{kukushkin2013increasing}, there is an increasing selection $r:f(L)\to L$ of $F$.
	
	For any $t\le t'$ in $f(L)$, we have $r(t')\in F(t')\subset F(t)$. Hence, $\{r(t')\}_{t'\ge t}$ is a chain in $F(t)$. Since $F(t)$ is chain-subcomplete  in $L$, the element $m:=\sup_Lr(f(L))=\sup_L\{r(t')\}_{t'\ge t}$ exists and is in $F(t)$. Then $f(m)\ge t$ for every $t\in f(L)$. Hence $m\in \argmax f$.	
	
	For any $u,v\in \argmax f$, since $f(u\vee v)\le \max f=f(u)$ (resp. $f(u\wedge v)\le \max f=f(u)$) and $f$ is quasisupermodular, one has $\max f=f(v)\le f(u\wedge v)$ (resp. $\max f=f(v)\le f(u\vee v)$). Therefore, $u\wedge v$ (resp. $u\vee v$) is in $\argmax f$. So, $\argmax f$ is a sublattice of $L$. Then by  Lemma \ref{lm:Sha3}, the nonempty sublattice $\argmax f$ is subcomplete.\end{proof}
\begin{rk}Zhou's theorem \cite[Thm.~1]{zhou1994set} is stated for  correspondences whose values are subcomplete sublattices. That is why we need Theorem \ref{thm:Kuku}  in the proof of Theorem \ref{thm:order}. \end{rk}
\section{Nash equilibria of quasisupermodular games}\label{sec:qsgame}
In  Section \ref{sec:qsgame}, we recall the definition of quasisupermodular games and give several results about the existence and structure of Nash equilibria. Quasisupermodular games retain the main feature of supermodular games, i.e., complementarities among actions.  The scope of quasisupermodular games surpasses that of supermodular games, broadening the spectrum of economic scenarios to which it can be applied.
\begin{df}\cite[p.160]{milgrom1994monotone}\label{df:single}
	Let $P,Q$ be two posets. Let $C$ be a chain. A function $f:P\times Q\to C$ satisfies the \emph{single crossing property} relative to $(P,Q)$ if for any $p\le p'$ in $P$ and any $q\le q'$ in $Q$, the condition $f(p,q)\le f(p',q)$ implies $f(p,q')\le f(p',q')$, and the condition $f(p,q)< f(p',q)$ implies $f(p,q')< f(p',q')$. 
\end{df}

A function having increasing difference relative to $P\times Q$ (in the sense of \cite[p. 42]{topkis1998supermodularity}) satisfies the single crossing property relative to $(P,Q)$. However, unlike the increasing difference property, single crossing property is not symmetric in the two variables, as shown by \cite[Example~2.6.9]{topkis1998supermodularity}.

Lemma \ref{lm:pairquasisupermodular} compares the hypotheses of Theorem \ref{thm:order} with those in \cite[Ch.~10, Theorem 2]{veinott1992lattice}.
\begin{lm}\label{lm:pairquasisupermodular}
	Let $L,L'$ be two lattices. Let $C$ be a chain.  Let $f:L\times L'\to C$ be a function. Suppose that for every chain $\cC\subset L'$, the restriction $f|_{L\times \cC}:L\times \cC\to C$ is  quasisupermodular. Then $f$ satisfies the single crossing property relative to $(L,L')$, and for every $t\in L'$, the function $f(\cdot,t):L\to C$ is quasisupermodular.
\end{lm}
\begin{proof}
	Consider   $x\le x'$ in $L$ and  $t\lneq t'$ in $L'$. 
	Then $\{t,t'\}$ is a chain in $L'$. Thus, $f|_{L\times \{t,t'\}}$ is quasisupermodular.
	Since  $(x,t')\wedge (x',t)=(x,t)$ and $(x,t')\vee (x',t)=(x',t')$, the condition  $f(x,t)\le f(x',t)$ (resp. $f(x,t)< f(x',t)$) implies $f(x,t')\le f(x',t')$ (resp. $f(x,t')< f(x',t')$). Thus, $f$ satisfies the single crossing property relative to $(L,L')$. Since $\{t\}$ is a chain in $L'$, the function $f|_{L\times \{t\}}:L\times \{t\}\to C$ is quasisupermodular, i.e., $f(\cdot,t):L\to C$ is quasisupermodular.\end{proof} 

\begin{eg}\label{eg:Veinott}
	Let $L=\{0,1,a,b\}$ be a lattice, where  $0=\min L$, $1=\max L$ and $a,b$ are incomparable. Then $a\wedge b=0$ and $a\vee b=1$. Let $C=\{0,1\}$ be a \emph{chain}.  Define a function  $f:L\times C\to \R$ by \begin{gather*}f(0,0)=f(b,1)=2,\, f(a,0)=f(b,0)=1,\\
		f(1,0)=f(1,1)=0,\, f(a,1)=9,\, f(0,1)=10.\end{gather*}
	Then $f$ satisfies the single crossing property relative to $(L,C)$. For every $t\in C$, the  function $f(\cdot,t):L\to \R$ is quasisupermodular.  However, $f$ is not quasisupermodular. Indeed, one has $f((a,0)\wedge (b,1))=f(0,0)=2=f(b,1)$ and $f((a,0)\vee (b,1))=f(1,1)=0<1=f(a,0)$.
\end{eg}
\begin{df}\label{df:qsgame}
	A quasisupermodular game \begin{equation}\label{eq:qsgame}(N,\{S_i\}_{i\in N},\{f_i\}_{i\in N},\{C_i\}_{i\in N})\end{equation} is the following data: \begin{enumerate}
		\item a nonempty  set of players $N$;
		\item for every $i\in N$, a nonempty lattice $S_i$ of the strategies of player $i$. Write $S=\prod_iS_i$ for the set of joint strategies and $S_{-i}:=\prod_{j\neq i}S_j$;
		\item for every $i\in N$, a nonempty chain $C_i$ representing the possible gains of player $i$;
		\item for each player $i\in N$, a payoff function $f_i: S\to C_i$  such that for every $x_{-i}\in S_{-i}$, the function $f_i(\cdot,x_{-i}):S_i\to C_i$ is quasisupermodular  (Definition \ref{df:qsi});
		\item for every $i\in N$, the function $f_i$ satisfies the single crossing property relative to $(S_i,S_{-i})$.
	\end{enumerate}
\end{df}\begin{rk}Such a game bears various names in the literature. It is called a game with the single crossing property in \cite[Sec.~3]{shannon1990ordinal} and
	is said to have (ordinal) strategic complementarities in \cite[p.175]{milgrom1994monotone}.  The definition on \cite[p.179]{topkis1998supermodularity} requires further   each $S_i$ to be a sublattice of some Euclidean space.\end{rk} 
From now to the end of the paper, fix a quasisupermodular game (\ref{eq:qsgame}).

\begin{df}[(Pure) Nash equilibrium]
	A  joint strategy $x\in  S$ is a Nash equilibrium of the game (\ref{eq:qsgame}) if for every $i\in N$ and every $y_i\in S_i$, one has $f_i(y_i,x_{-i})\le f_i(x)$. The subset of $S$ comprised  of all Nash equilibria is denoted by $E$.
\end{df} \begin{rk} Although we cling to the analysis of pure-strategy Nash equilibria,  mixed-strategy Nash equilibria are discussed in \cite{echenique2003mixed} for quasisupermodular games.\end{rk}
Classically, the best response correspondence  is an essential tool to study Nash equilibria.	For each player $i\in N$, the (individual) best response correspondence  $R_i:S\to 2^{S_i}$ is defined by \[R_i(x)=\mathrm{argmax}_{y_i\in S_i}f_i(y_i,x_{-i}).\]The joint best response $R:S\to 2^S$ is defined as $R(x)=\prod_{i\in N}R_i(x_{-i})$. Then $E$ coincides with the set of fixed points of $R$.

Theorem \ref{thm:order} is purely order theoretic. For supermodular games, a purely order-theoretic result about the existence of the largest and the least Nash equilibria is in \cite[p.1266]{milgrom1990rationalizability}.
\begin{thm}\label{thm:order}
	Assume that for every $i\in N$, the lattice $S_i$ is complete, and for every $x_{-i}\in S_{-i}$, the function $f_i(\cdot,x_{-i}):S_i\to C_i$ is upper chain subcomplete. Then $E$ (with the order structure inherited from $S$) is a nonempty complete lattice.
\end{thm}
\begin{proof}
	By Theorem \ref{thm:Kuku}, for every $x\in S$ and every $i\in N$,  the value $R_i(x)\subset S_i$ is a nonempty subcomplete sublattice. Therefore, so is $R(x)\subset S$. Moreover, by \cite[Theorem 2.8.6]{topkis1998supermodularity}, the correspondence $R_i$ is increasing in the sense of \cite[p.33]{topkis1998supermodularity}. So $R:S\to 2^S$ is also increasing. Because every $S_i$ is complete, the product lattice $S$ is complete. Then by Zhou's  theorem \cite[Thm.~1]{zhou1994set}, $E$ is a nonempty complete lattice.
\end{proof}
\begin{rk}By Proposition \ref{pp:topo},  the continuity hypothesis of Theorem \ref{thm:order} is weaker than the topological upper semicontinuity. 
	
	Veinott \cite[Ch.~10, Thm.~2]{veinott1992lattice} gives a similar result for parameterized games.  When restricted to normal form games, Veinott's theorem requires that for every $i\in N$ and every chain $\cC\subset S_{-i}$, the function $f_i|_{S_i\times \cC}:S_i\times \cC\to C$ is quasisupermodular. By Lemma \ref{lm:pairquasisupermodular}, this hypothesis implies that the game is a quasisupermodular game in the sense of Definition \ref{df:qsgame}. However, as Example \ref{eg:Veinott} shows, a quasisupermodular game may not satisfy Veinott's hypothesis. Furthermore, Veinott's theorem requires  every $C_i$ to be complete.   In this sense, Theorem \ref{thm:order} is not covered by Veinott's theorem.\end{rk}

\begin{eg}
	Let $N=\{1,2\}$. Let $C_i=\R$ for $i=1,2$. Let $S_1$ be the complete lattice $L$ constructed in Example \ref{eg:A3}. Let $S_2$ be the finite lattice $X$ in Example \ref{eg:Veinott}. Let $x_0,f$ be as in  Example \ref{eg:A3}. Define a function $f_1:S_1\times S_2\to \R,\quad (s_1,s_2)\mapsto f(s_1)$.  Then $f_1(\cdot,0):S_1\to \R$  is not upper semicontinuous in the interval topology of $S_1$.  Define a function $f_2:S_1\times S_2\to \R$ by $f_2(s_1,0)=0$, $f_2(s_1,a)=f_2(s_1,b)=1$ and $f_2(s_1,1)=1.5$ for every $s_1\in S_1$. Then $f_2(s_1,\cdot):S_2\to \R$ is  not supermodular. Still, the corresponding game satisfies the hypotheses of Theorem \ref{thm:order}. In this case, the set $E=(L\setminus\{x_0\})\times \{1\}$.
\end{eg}

For every $i\in N$, assume $C_i=\R$. Let $\tau_i$ be a \emph{compact} topology on $S_i$ finer than the interval topology of $S_i$.  By the Frink-Birkhoff theorem \cite[Theorem 20, p.250]{birkhoff1940lattice},  the lattice $S_i$ is complete. Consequently, their product $S$ is also a complete lattice.

Theorem \ref{thm:topo+order} below is similar to \cite[Theorem 1]{prokopovych2017strategic}. In  \cite[Theorem 1]{prokopovych2017strategic},  each player's strategy set $S_i$ is a  chain, the set of players $N$ is finite and only the existence of Nash equilibria is proved.  On the one hand, we allow the $S_i$ to be  lattices, the set $N$ to be infinite and the existence of the greatest Nash equilibrium is established.  On the other hand, \cite{prokopovych2017strategic} weakens the single crossing condition (Definition \ref{df:qsgame}) to upward or downward transfer single-crossing condition  (\cite[Definition 4]{prokopovych2017strategic}).
\begin{thm}\label{thm:topo+order}
	Assume that	 for every $i\in N$ and every $x_{-i}\in S_{-i}$, the function $f_i(\cdot,x_{-i}):S_i\to \R$ is transfer weakly upper continuous relative to $\tau_i$, and upward (resp. downward) upper semicontinuous.  Then there  is  a largest (resp. least) Nash equilibrium. 
\end{thm}
\begin{proof}By symmetry, it suffices to prove the statement without parentheses. For every $i\in N$ and every $x_{-i}\in S_{-i}$, by  \cite[Theorem 1]{tian1995transfer}, the compactness of $\tau_i$ and the transfer weak upper continuity of $f_i(\cdot,x_{-i})$ imply that $R_i(x)$ is nonempty. Therefore, by \cite[Thm.~2.8.6]{topkis1998supermodularity}, the correspondence $R_i:S\to 2^{S_i}$ is increasing. So, the  correspondence  $R:S\to 2^S$ is increasing. Whence, $R(x)$ is a sublattice of $S$ for every $x\in S$.
	
	We claim that for every $x\in S$, the largest element of $R(x)$ exists. 
	
	For every $i\in N$ and every chain $C$ in $R_i(x)$, the upward upper semicontinuity implies \[f(\sup_{S_i}(C),x_{-i})\ge \limsup_{t\in C,t\uparrow\sup_{S_i}(C)}f(t,x_{-i})=\max_{S_i} f(\cdot,x_{-i}),\] so $\sup_{S_i}(C)\in R_i(x)$. Now that every chain in the nonempty poset $R_i(x)$ has an upper bound, by Zorn's lemma, $R_i(x)$ has a maximal element $m_i$. Then $(m_i)_i\in R(x)$ is a maximal element of $R(x)$. By Lemma \ref{lm:min}, it is the largest element of $R(x)$. The claim is proved.
	
	As $R$ is increasing, the single-valued function $\max R:S\to S$ is increasing. As $S$ is a complete lattice, $\max S$ exists. Then by Tarski's fixed point theorem \cite[Thm.~1]{tarski1955lattice}, the function $\max R:S\to S$ has a fixed point $x^*:=\max\{y\in S:y\le \max R(y)\}$. Then $x^*$ is a Nash equilibrium. 
	
	We show that $x^*$ is the largest Nash equilibrium. If $s\in S$ is another Nash equilibrium, then $s\in R(s)\le \max R(s)$. By definition of $x^*$, we have $s\le x^*$, which means that $x^*$ is largest Nash equilibrium. 
\end{proof}
\begin{eg}
	Let $N=\{1,2\}$. Let $S_1=[0,1]$ and $S_2=[0,1]^2$. Define a function $f_1:S_1\times S_2=[0,1]^3\to \R$ by \[(s_1,a,b)\mapsto \begin{cases}
		1 & s_1>0,\\
		0 & s_1=0.
	\end{cases}\] Then $f_1(\cdot,0,0):S_1\to \R$ is not  order upper semicontinuous. Define a function $f_2:S_1\times S_2=[0,1]^3\to \R,\quad (s_1,a,b)\mapsto s_1+a+b$. Then the resulting game satisfies the conditions of Theorem \ref{thm:topo+order}. The set $E$ is $(0,1]\times \{(1,1)\}$. The largest Nash equilibrium is $(1,1,1)$, but there is no least Nash equilibrium. As $S_2$ is not a chain, one cannot apply \cite[Theorem 1]{prokopovych2017strategic} in this case.
\end{eg}
With a topological hypothesis (on the payoff functions)  strictly weaker than  continuity, we  get the same order structure of Nash equilibria as in Theorem \ref{thm:order}.
\begin{thm}\label{thm:topo}
	Assume that for every $i\in N$ and every $x_{-i}\in S_{-i}$, the function $f_i(\cdot,x_{-i}):S_i\to \R$ is transfer upper continuous relative to $\tau_i$.  Then $E$  is a nonempty complete lattice.
\end{thm}
\begin{proof} For every $i\in N$ and every $x\in S$, by Fact \ref{ft:Tian} and the compactness of $\tau_i$,   the value $R_i(x)$ is nonempty and compact in $(S_i,\tau_i)$. As $\tau_i$ is finer than the interval topology of $S_i$, the subset $R_i(x)$ is  compact in the interval topology.
	By \cite[Theorem 2.8.6]{topkis1998supermodularity}, the correspondence $R_i:S\to 2^{S_i}$ is  increasing. Therefore, the correspondence $R:S\to 2^S$ is increasing and $R_i(x)$ is a sublattice of $S_i$. By \cite[Topkis's theorem]{yu2022existence}, the sublattice $R_i(x)$ is subcomplete in $S_i$. Therefore, $R(x)$ is a nonempty subcomplete sublattice of $S$.
	By Zhou's theorem \cite[Thm.~1]{zhou1994set}, $E$ is a nonempty complete lattice. 
\end{proof}

Corollary \ref{cor:Calciano} follows from Theorem \ref{thm:topo}.  It specializes to Calciano's result in \cite[Sec 4.1]{calciano2007games} about the structure of the set of  Nash equilibria when $N$ is finite. 
\begin{cor}[Calciano]\label{cor:Calciano}Suppose that for every $i\in N$ and every $x_{-i}\in S_{-i}$, the function $f_i(\cdot,x_{-i}):S_i\to\R$ is  upper semicontinuous relative to $\tau_i$. Then  $E$  is a nonempty  complete lattice.
\end{cor}

\begin{rk}
	Under an extra continuity hypothesis  (namely, for every $i\in N$ and every $x_i\in S_i$, the function $f_i(x_i,\cdot):S_{-i}\to \R$ is continuous), the existence of both least and largest Nash equilibria is  showed in \cite[Theroem 12]{milgrom1994monotone}, extending a previous result \cite[Proposition 4]{shannon1990ordinal}.
\end{rk}
\begin{eg}
	Let $N=\{1,2\}$. Let $S_1=S_2=[0,1]$. Define  a function $f_1:S_1\times S_2\to \R$ by \[(s_1,s_2)\mapsto\begin{cases}
		0 & s_1\le 1/3,\\
		1 & 1/3<s_1<2/3,\\
		2 & 2/3\le s_1.
	\end{cases}\] Then $f_1(\cdot,0):S_1\to \R$ is not upper semicontinuous. Define $f_2:S_1\times S_2\to \R,\quad (s_1,s_2)\mapsto (s_1+1)(s_2-s_2^2)$. The resulting game satisfies the conditions of Theorem \ref{thm:topo} but not those of Corollary \ref{cor:Calciano}. The set $E=[2/3,1]\times \{1/2\}$. Thus, Theorem \ref{thm:topo} is strictly stronger than Corollary \ref{cor:Calciano}.
\end{eg}
\begin{cor}
	If for every $i\in N$, the set $S_i$ is finite, then  $E$ is a nonempty complete  lattice.
\end{cor}
\begin{proof}For every $i\in N$, take $\tau_i$ to be the discrete topology on $S_i$. Then $\tau_i$ is compact, and for every  $x_{-i}\in S_{-i}$,  the function $f_i(\cdot,x_{-i}):S_i\to \R$ is  upper semicontinuous relative to $\tau_i$. The result follows from Corollary \ref{cor:Calciano}.
\end{proof}
\subsection*{Acknowledgments}
	I express gratitude to  Philippe Bich, my  supervisor, for his invaluable feedback and  much support. I am sincerely grateful to the two anonymous referees, for their detailed comments. I thank them for pointing out \cite{davey2002introduction,roman2008lattices,durieu2008ordinal}, and for suggesting the generalization of Propositions \ref{pp:usc} and \ref{pp:topo} from real-valued functions to chain-valued functions.
	
	\bibliography{ref.bib}
	\bibliographystyle{alpha}
\end{document}